\documentclass[manuscript,nonacm]{acmart}

\usepackage{hyperref}
\usepackage{blindtext}

\usepackage{xpatch}

\makeatletter
\newcommand{\institute}[1]{\newcommand{\institute}{#1}}
\xpatchcmd{\ps@firstpagestyle}{Manuscript submitted to ACM}{}{\typeout{}}{\typeout{first patch failed}}
\xpatchcmd{\ps@standardpagestyle}{Manuscript submitted to ACM}{}{\typeout{Second patch succeeded}}{\typeout{d}}    \@ACM@manuscriptfalse
\let\@authorsaddresses\@empty
\renewcommand\@formatdoi[1]{\ignorespaces}
\makeatother

\renewcommand\footnotetextcopyrightpermission[1]{} 
\setcopyright{none}
\usepackage{graphicx}

\usepackage{mathtools}
\DeclarePairedDelimiter\ceil{\lceil}{\rceil}

\newtheorem{observation}{Observation}

\begin{document}

\title{Distributed Computation and Reconfiguration in Actively Dynamic Networks}

\author{Othon Michail}
\email{}
\affiliation{
	\institution{Department of Computer Science, University of Liverpool, UK \newline Othon.Michail@liverpool.ac.uk}
	\country{}	
}

\author{George Skretas}
\email{}
\affiliation{
  \institution{Department of Computer Science, University of Liverpool, UK \newline  G.Skretas@liverpool.ac.uk}
  \country{}
}

\author{Paul G. Spirakis}
\affiliation{
  \institution{Department of Computer Science, University of Liverpool, UK and Computer Engineering and Informatics Department, University of Patras, Greece \newline P.Spirakis@liverpool.ac.uk}
  \country{}
}

\renewcommand{\shortauthors}{}
\begin{abstract}
In this paper, motivated by recent advances in the algorithmic theory of dynamic networks, we study systems of distributed entities that can actively modify their communication network. This gives rise to distributed algorithms that apart from communication can also exploit network reconfiguration in order to carry out a given task. At the same time, the distributed task itself may now require a global reconfiguration from a given initial network $G_s$ to a target network $G_f$ from a family of networks having some good properties, like small diameter. With reasonably powerful computational entities, there is a straightforward algorithm that transforms any $G_s$ into a spanning clique in $O(\log n)$ time, where time is measured in synchronous rounds and $n$ is the number of entities. From the clique, the algorithm can then compute any global function on inputs and reconfigure to any desired target network in one additional round.

We argue that such a strategy, while time-optimal, is impractical for real applications. In real dynamic networks there is typically a cost associated with creating and maintaining connections. To formally capture such costs, we define three reasonable edge-complexity measures: the \emph{total edge activations}, the \emph{maximum activated edges per round}, and the \emph{maximum activated degree of a node}. The clique formation strategy highlighted above, maximizes all of them. We aim at improved algorithms that will achieve (poly)log$(n)$ time while minimizing the edge-complexity for the general task of transforming any $G_s$ into a $G_f$ of diameter (poly)log$(n)$.

There is a natural trade-off between time and edge complexity. Our main lower bound shows that $\Omega(n)$ total edge activations and $\Omega(n/\log n)$ activations per round must be paid by any algorithm (even centralized) that achieves an optimum of $\Theta(\log n)$ rounds. On the positive side, we give three distributed algorithms for our general task. The first runs in $O(\log n)$ time, with at most $2n$ active edges per round, an optimal total of $O(n\log n)$ edge activations, a maximum degree $n-1$, and a target network of diameter 2. The second achieves bounded degree by paying an additional logarithmic factor in time and in total edge activations, that is, $O(\log^2 n)$ and $O(n\log^2 n)$, respectively. It gives a target network of diameter $O(\log n)$ and uses $O(n)$ active edges per round. Our third algorithm shows that if we slightly increase the maximum degree to polylog$(n)$ then we can achieve a running time of $o(\log^2 n)$. 

This novel model of distributed computation and reconfiguration in actively dynamic networks and the proposed measures of the edge complexity of distributed algorithms may open new avenues for research in the algorithmic theory of dynamic networks. At the same time, they may serve as an abstraction of more constrained active-reconfiguration systems, such as reconfigurable robotics which come with geometric constraints, and draw interesting connections with alternative network reconfiguration models, like overlay network construction and network constructors. We discuss several open problems and promising future research directions.\\

\end{abstract}


\maketitle
\thispagestyle{empty}
\section{Introduction}

\subsection{Dynamic Networks}

The \emph{algorithmic theory of dynamic networks} is a relatively new area of research, concerned with studying the algorithmic and structural properties of networked systems whose structure changes with time. 

One way to classify dynamic networks is based on \emph{who controls the network dynamics}. In \emph{passively} dynamic networks the changes are external to the algorithm, in the sense that the algorithm has no control over them. Such dynamics are usually modeled by sequences of events determined by an \emph{adversary scheduler}. This is for example the case when the computing entities must operate in a dynamic environment, such as when being carried by a set of transportation units. In other applications, the entities can \emph{actively} control the dynamics of their network, as is the case in mobile or reconfigurable robotics and peer to peer networks. \emph{Hybrid} cases or cases of \emph{partial control} are less studied (cf. \cite{GHSS17} for a relevant study).

Another level of classification comes from \emph{who controls the algorithm}. This gives rise to two main families of models. One is the \emph{fully centralized}, in which a central controller has global view of the system. In case of active network dynamics, the centralized algorithm typically designs a dynamic network by exploiting its full knowledge about the system in a way that aims to optimize some given objective function. If network dynamics are passive then the goal is typically to achieve some global computation task, like foremost journeys or dissemination, which may either be possible to compute \emph{offline} under full information about the evolution of the network or required to compute \emph{online} under limited or no knowledge about the future network structure. Similar objectives hold for the \emph{fully distributed} case, in which every node in the network is an independent computing entity, like an automaton or Turing machine, typically equipped with computation and communication capabilities, and in the case of active dynamics with the additional capability to locally modify the network structure, like \emph{activating} a connection to a new neighbor or \emph{eliminating} an existing connection. One may also consider \emph{partial distributed control}, in which only $k$ out of $n$ nodes are occupied by computing entities, but again not much is known about the latter family of models. 

\subsection{An Actively Dynamic Distributed Model}

In this paper, we consider an \emph{actively dynamic fully distributed} system. In particular, there are $n$ computing entities starting from an \emph{initial connected network} drawn from a family of initial networks. The entities are typically equipped with unique IDs, can compute locally, can communicate with neighboring entities, and can activate connections to new neighbors locally or eliminate some of their existing connections. All these take place in lock step through a standard synchronous message passing model, extended to include the additional operations of edge activations and deactivations within each round.

The goal is, generally speaking, to program all the entities with a distributed algorithm that can transform the initial network $G_s$ into a \emph{target network} $G_f$ from a family of target networks. The idea is that starting from a $G_s$ not necessarily having a good property, like small diameter, the algorithm will be able to ``efficiently'' reach a $G_f$ satisfying the property. This gives rise to two main objectives, which in some cases it might be possible to satisfy at the same time. One is to transform a given $G_s$ into a desired target $G_f$ and the other is to exploit some good properties of $G_f$ in order to more efficiently solve a distributed task, like computation of a global function through information dissemination. 

Even when edge activations are extremely local, meaning that an edge $uv$ can only be activated if there exists a node $w$ such that both $uw$ and $wv$ are already active, there is a straightforward algorithmic strategy that can successfully carry out most of the above tasks. In every round, all nodes activate all of their possible new connections, which corresponds to each node $u$ connecting with all nodes $v_i$ that were at distance 2 from $u$ in the beginning of the current round. By a simple induction, it can be shown that in any round $r$ the neighborhood of every node has size at least $2^r$, which implies that a spanning clique $K_n$ is formed in $O(\log n)$ rounds. Such a clique can then be used for global computations, like electing the maximum id as a leader, or for transforming into any desired target network $G_f$ through eliminating the edges in $E(K_n)\setminus E(G_f)$. All these can be performed within a single additional round.

Even though sublinear global computation and network-to-network transformations are in principle possible through the \emph{clique formation} strategy described above, this algorithmic strategy still has a number of properties which would make it impractical for real distributed systems. As already highlighted in the literature of dynamic networks, activating and maintaining a connection does not come for free and is associated with a cost that the network designer has to pay for. Even if we uniformly charge 1 for every such active connection, the clique formation incurs a cost of $\Theta(n^2)$ total edge activations in the worst case and always produces instances (e.g., when $K_n$ is formed) with as many as $\Theta(n^2)$ active edges in which all nodes have degree $\Theta(n)$.

Our goal in this work is to formally define such cost measures associated with the structure of the dynamic network and to give improved algorithmic strategies that maintain the time-efficiency of clique formation, while substantially improving the edge complexity as defined by those measures. In particular, we aim at minimizing the edge complexity, given the constraint of (poly)logarithmic running time. Observe at this point that without any restriction on the running time, a standard distributed dissemination solely through message passing over the initial network, would solve global computation without the need to activate any edges. However, linear running times are considered insufficient for our purposes (even when the goal is to solve traditional distributed tasks). Moreover, strategies that do not modify the input network cannot be useful for achieving network-to-network transformations.

\subsection{Contribution}

We define three cost measures associated with the edge complexity of our algorithms. One is the \emph{total number of edge activations} that the algorithm  performed during its course, the second one is the \emph{maximum number of activated edges in any round} by the algorithm, and the third one is the \emph{maximum activated degree of a node in any round}, where the maximum activated degree of a node is defined only by the edges that have been activated by the algorithm.

Our ultimate goal in this paper is to give \emph{(poly)logarithmic time} algorithms which, starting from \emph{any} connected network $G_s$, transform $G_s$ into a $G_f$ of \emph{(poly)logarithmic diameter} and at the same time \emph{elect a unique leader}.
Such algorithms can then be composed with any algorithm $B$ that assumes an initial network of (poly)logarithmic diameter and has access to a unique leader and unique ids. In case of a static network algorithm $B$, this for example yields (poly)logarithmic time information dissemination and computation of any global function on inputs. In case of an actively dynamic network algorithm $B$, it gives (poly)logarithmic time transformation into any target network from a given family which depends on restrictions related to the edge complexity. 

We restrict our focus on \emph{deterministic} algorithms, that is, the computational entities do not have access to any random choices. Moreover, our algorithms never break the connectivity of the network of active edges as this would result in components that could never be reconnected based on the permissible edge activations. Even though this is in principle permitted, it cannot be useful for the small diameter and spanning target networks that we are aiming for in this work. Temporary disconnections within a round may be permitted but can always be avoided by first activating all new edges and then deactivating any edges for the current round.  

To appreciate the difficulty in solving the above problem 
while optimizing the edge complexity, assume for a moment, a network as simple as a spanning line $u_0u_1$ $\cdots u_{n-1}$ with a pre-elected unique leader on one of its endpoints, say $u_0$. If we had global view of the system, then we would proceed in $\log n$ phases as follows. In every phase $i$, we would start from $u_0$ and activate edges by making hops of length $2$ over the edges activated in the previous phase, thus, activating the edges $u_0u_{2^i}, u_{2^i}u_{2\cdot 2^i}, u_{2\cdot 2^i}u_{3\cdot 2^i},\ldots$ in the current phase. This would give an edge for every $2^i$ consecutive nodes in phase $i$ and a total of $O(n)$ edge activations. The diameter of the resulting network and the number of phases are both logarithmic in $n$. Observe now that this basic construction essentially requires to determine which half of the nodes that activated an edge in the previous phase will be the ones to activate in the current phase. But all these nodes are bound to behave identically given an order-equivalence of received ids in their local history and there is no obvious way to exploit the pre-elected leader at $u_0$ for symmetry breaking, as its initial distance from many of them is asymptotically equal to the original diameter of the network, that is, $\Theta(n)$. What this example reveals, is an inherent trade-off between time and edge activations stemming from the inability of the distributed entities to break symmetry locally and, thus, fast. Intuitively, breaking symmetry takes time and, if left unbroken, costs many edge activations every time one of the nodes decides to activate. 

The difficulties that we just highlighted are formally captured by our lower bounds presented in Section \ref{SEC:LOWER-BOUNDS}. In particular, we first prove that $\Omega(\log n)$ is a lower bound on time following from an upper bound of 2 on the distance of new connections and the $\Theta(n)$ worst-case diameter of the initial network. Then we give an $\Omega(n)$ lower bound on total edge activations and $\Omega(n/\log n)$ activations per round for any centralized algorithm that achieves an optimal $\Theta(\log n)$ time. Our main lower bound is a total of $\Omega(n\log n)$ edge activations that any logarithmic time distributed algorithm must pay. This is in contrast to the $\Theta(n)$ total edges that would be sufficient for a centralized algorithm and is due to the distributed nature of the systems under consideration.    

We begin our algorithmic constructions with some basic algorithms for special types of initial and target networks, which will then be used as core components in our general algorithms. These are discussed in Section \ref{SUBSEC:BASIC}. One of these algorithms transforms any rooted tree into a star and the other an oriented spanning line into a complete binary tree. Both operate in $O(\log n)$ time, have a linear number of active edges per round and an optimal total of $O(n\log n)$ edge activations. The latter algorithm additionally maintains a maximum degree of at most 3 throughout its course, while the degree of the former is necessarily linear, due to ending up in a spanning star.

We then proceed to our main positive results. In particular, we give three algorithms for transforming \emph{any} initial connected network $G_s$ into a network $G_f$ of (poly)logarithmic diameter and at the same time electing a unique leader. Each of these algorithms makes a different contribution to the time vs edge complexity trade-off. All of our main algorithms are built upon the following general strategy. For each of them, we define a different \emph{gadget network} and the algorithms are developed in such a way that they always satisfy the following invariants. In any round of an execution, the network is the union of committees being such gadget networks of varying sizes and some additional edges including the initial edges and other edges used to join the committees. Initially, every node forms its own committee and the algorithms progressively merge pairs or larger groups of committees based on the rule that the committee with the greater id dominates. If properly performed, this ensures that eventually only one committee remains, namely, the committee of the node $u_{max}$ with maximum id in the network. The diameter of all our gadgets is (poly)logarithmic in their size, which facilitates quick merging and ensures that the final committee of $u_{max}$ satisfies the (poly)$\log(n)$ diameter requirement for $G_f$. The algorithms also ensure that, by the time the committee of $u_{max}$ is the unique remaining committee, $u_{max}$ is the unique leader elected. 

Our algorithms must achieve (poly)logarithmic time and they do so by satisfying the invariant that winning committees always grow exponentially fast. This growth is \emph{asynchronous} in our algorithms for the following reason. In a typical configuration (of a phase) the graph of mergings forms a spanning forest $F$ of committees such that any tree $T$ in $F$ is rooted at the committee that will eventually consume all committees in $V(T)$. Given that those trees may have different sizes (even up to $V(T)=\Theta(n)$), the winning times of different committees may be different, but we can still show that their amortized growth is exponential. 

Our first algorithm, called GraphToStar and presented in Section \ref{sec:algorithm-1}, uses a star network as a gadget. Its running time is $O(\log n)$ and it uses at most $2n$ active edges per round and an optimal total of $O(n\log n)$ edge activations. The target network $G_f$ that it outputs is a spanning star, thus, the algorithm achieves a final diameter of 2. 

Our second algorithm, called GraphToWreath and presented in Section \ref{SEC:ALGORITHM-2}, uses as a gadget a type of graph, which we call a \emph{wreath} and is the union of a ring and a complete binary tree spanning the ring. The main improvement compared to GraphToStar is that it maintains a bounded maximum degree throughout its course (given a bounded-degree $G_s$). It does this at the cost of increasing the running time to $O(\log^2 n)$ and the number of total edge activations to $O(n\log^2 n)$. The active edges per round remain $O(n)$. The target network $G_f$ that it outputs is a spanning complete binary tree (after deleting the original edges and the spanning ring), thus, the algorithm achieves a final diameter of $O(\log n)$. 

Our third algorithm, called GraphToThinWreath and presented in Section \ref{SEC:ALGORITHM-3}, shows that if we slightly increase the maximum degree to polylog$(n)$ then we can achieve a running time of $o(\log^2 n)$ (more precisely, $O(\log^2 n/\log\log^k n)$, for some constant $k\geq 1$).

If our model can be compared to models from the area of overlay networks construction (see Section \ref{subsec:related-work} for a discussion on this matter), then GraphToWreath is, to the best of our knowledge, the first deterministic bounded-degree $O(\log^2 n)$-time algorithm and GraphToThinWreath the first deterministic polylog$(n)$-degree $o(\log^2 n)$-time algorithm for the problem of transforming any connected $G_s$ into a polylog$(n)$ diameter $G_f$.

\subsection{Related Work}
\label{subsec:related-work}

\noindent\textbf{Temporal Graphs.} The algorithmic study of temporal graphs was initiated by Berman \cite{Be96} and Kempe \emph{et al.} \cite{KKK00}, who studied a special case of temporal graphs in which every edge can be available at most once. The problem of designing a cost-efficient temporal graph satisfying some given connectivity properties was introduced in \cite{MMS19}. The design task was carried out by an offline centralized algorithm starting from an empty edge set. Subsequent work \cite{EMMZ19}, motivated by epidemiology applications, considered the centralized algorithmic problem of re-designing a given temporal graph through edge deletions in order to end up with a temporal graph with bounded temporal reachability, thus keeping the spread of a disease to a minimum. Our work is related to the temporal network (re-)design problem but our model is fully distributed, allows for both edge activations and deletions, and our families of target networks are different than those considered in the above papers.

\noindent\textbf{Distributed Computation in Passively Dynamic Networks.} Probably the first authors to consider distributed computation in passively dynamic networks were Angluin \emph{et al.} \cite{AADFP06,AAER07,AR09}. Their population protocol model, considered originally the computational power of a population of $n$ finite automata which interact in pairs passively either under an eventual fairness condition or under a uniform random scheduling assumption. A variant of population protocols in which the automata can additionally create or destroy connections between them was introduced in \cite{MCS11-2,MS16a}. It was shown that in that model, called network constructors, complex spanning networks can be created efficiently despite the computational weakness of individual entities. The closest to our approach from this area is \cite{MS17a}, in which the authors showed how to transform any connected initial network into a spanning line which can then be exploited to achieve global computation on input values and termination. The main difference though is that in all these models pairwise interactions are chosen asynchronously by a scheduler, and connections can be created between any pair of nodes during their interaction independently of the current network structure and the distance between them. 

Other papers \cite{OW05,KLO10,MCS13b} have studied distributed computation in worst-case dynamic networks using a traditional message-passing model and typically operating through local broadcast in the current neighborhood. Our communication model is closer to those models but network dynamics there are always passive and their main goal has been to revisit the complexity of classical distributed tasks under a worst-case adversarial network.

\noindent\textbf{Construction of Overlay Networks.} There is a rich literature on the distributed construction of overlay networks. A typical assumption is that there is an overlay (active) edge from a node $u$ to a node $v$ in a given round iff $u$ has obtained $v$'s id through a message. 
Without further restrictions, the overlay in round $r$ would always correspond to the union of $r$ consecutive transitive extensions starting from the original edge set. The main restriction imposed in the relevant literature is a polylogarithmic (in bits) communication capacity per node per round, which also implies that in every round $O(\log n)$ new overlay connections per node are permitted. 

Our model and results, even though different in motivation, in the complexity measures considered, and in the restrictions we impose, appear to have similarities with some of the developments in this area. Unlike our work, where our complexity measures are motivated by the cost of creating and maintaining physical or virtual connections, the algorithmic challenges in overlay networks are mainly due to restricting the communication capacity of each node per round to a polylogarithmic total number of bits.

Research in this area started with seminal papers such as Chord of Stoica \emph{et al.} \cite{SMK01} and the Skip graphs of Aspnes and Shah \cite{AS07}. Probably the first authors to have considered the problem of constructing an overlay network of logarithmic diameter were Angluin \emph{et al.} \cite{AACW05}. Their algorithm is randomized and has a running time of $O((d +W) \log n)$ w.h.p., where $W$ is the maximum size of a unique id. Then Aspnes and Wu \cite{AW07} gave a randomized $O(\log n)$ time algorithm for the special case in which the initial network has outdegree at most 1. 

To the best of our knowledge, the only previous deterministic algorithm for the problem is the one by Gmyr \emph{et al.} \cite{GHSS17}. Our algorithmic strategies appear to have some similarities to their ``Overlay Construction Algorithm'', which in their work is used as a subroutine for monitoring properties of a passively dynamic network. Unlike our model, their model is hybrid in the sense that algorithms have partial control over the connections of an otherwise passively dynamic network. Due to using different complexity measures and restrictions it is not totally clear to us yet whether a direct comparison between them would be fair. Still, we give some first observations. Their algorithm has the same time complexity, i.e., $O(\log^2 n)$, with our GraphToWreath algorithm, while our GraphToStar algorithm achieves $O(\log n)$ and our GraphToThinWreath $o(\log^2 n)$. Their overlays appear to maintain $\Theta(n\log n)$ active connections per round, while our algorithms maintain $O(n)$. Their maximum active degree is polylogarithmic, the same as GraphToThinWreath, while GraphToStar uses linear and GraphToWreath always bounded by a constant. Their model restricts the communication capacity of every node to a polylogarithmic number of bits per round, whereas we do not restrict communication. 

A very recent work by G{\"o}tte \emph{et al.} \cite{GHS19} has improved the upper bound of \cite{AACW05} to $O(\log^{3/2} n)$, w.h.p. It is a randomized algorithm which uses a core deterministic procedure that has some similarities to our algorithmic strategy of maintaining and merging committees (called ``supernodes'' there) whose size increases exponentially fast. Their model keeps the polylogarithmic restriction on communication and the polylogarithmic maximum degree.  

Scheideler and Setzer \cite{SS19} recently studied the (centralized) computational complexity of computing the optimum graph transformation and gave $\mathbf{NP}$-hardness results and a constant-factor approximation algorithm for the problem. 

\noindent\textbf{Programmable Matter.} There is a growing recent interest in studying the algorithmic foundations of systems that can change their physical properties through local reconfigurations \cite{FRRS16,DGRS16,MSS18,AADD19,AMP19}. A prominent such property is changing their shape. Typical examples of systems in this area are reconfigurable robotics, swarm robotics, and self-assembly systems \cite{BG11,MC15}. In most of these settings, modification of structure can be represented as a dynamic network, usually called \emph{shape}, with additional geometric restrictions coming from the shape and the local reconfiguration mechanism of the entities. The goal is then to transform a given initial shape into a desired target shape through a sequence of valid local moves. Our network transformation problem can be viewed as an non-geometric abstraction of these geometric transformation problems. Apart from being motivated by this area, we also hope that the abstract algorithmic principles of network reconfiguration might promote our understanding of the geometrically constrained cases. 

\vspace{-0.3cm}
\section {Preliminaries}
\label{sec:prel}

\subsection{The Model}
\label{subsec:model}

An actively dynamic network is modeled in this work by a temporal graph $D=(V,E)$, where  $V$ is a static set of $n$ nodes and $E\subseteq \binom{V}{2}\times\mathbb{N}$ is a set of undirected time-edges. In particular, $E(i) =\{e: (e,i)\in E\}$ is the set of all edges that are \emph{active} in the temporal graph at the beginning of round $i$. Since V is static, $E(i)$ can be used to define a snapshot of the temporal graph at round $i$, which is the static graph $D(i)=(V,E(i))$. 

The temporal graph $D$ of an execution is generated by local operations performed by the nodes of the network, starting from an initial graph $G_s=D(1)$. Throughout this paper, $G_s$ is assumed to be connected. A node $u$ can \emph{activate an edge} with node $v$ in round $i$, if $uv\not\in E(i)$ and there exists a node $w$ such that both $uw$ and $wv$ are active at the beginning of round $i$. A node $u$ can \emph{deactivate an edge} with node $v$ in round $i$, provided that $uv\in E(i)$. An active edge remains active indefinitely unless a node who is incident to that edge deactivates it. There is at most one active edge between any pair of nodes, that is multiple edges are not allowed. If a node attempts to activate an edge which is already active, the action has no effect and the edge remains active; similarly for deactivating inactive edges. Moreover, if a node $u$ decides to activate an edge with a node $v$ in round $i$ and $v$ decides to activate an edge with $u$ in the same round, then only one edge is activated between them. In case $u$ and $v$ disagree on their decision about edge $uv$, then their actions have no effect on $uv$. We define $E_{ac}(i)$ as the set of all edges that were activated in round $i$ and $E_{dac}(i)$ as the set of all edges that were deactivated in round $i$. Then $E(i+1)=(E(i)\cup E_{ac}(i))\setminus E_{dac}(i)$.

We define set $N_{1}^{i}(u)$ of node $u$, where $v\in N_{1}^{i}(u)$ iff $uv\in E(i)$ which means that set $N_{1}^{i}(u)$ contains the neighbors of node $u$ in round $i$. Additionally, set $N_{2}^{i}(u)$ of node $u$, where $w\in N_{2}^{i}(u)$ iff there exists $ v\in V$ s.t. $v\in N_{1}^{i}(u)$ and $v\in N_{1}^{i}(w)$ and $w\not\in N_{1}^{i}(u)$. That is, set $N_{2}^{i}(u)$ of node $u$ in round $i$ contains the nodes at distance 2 which we will refer to as \emph{potential neighbors}. We will omit the $i$ index for rounds, when clear from context.

Each node $u\in V$ is identical to every other node $v$ but for the unique identifier \emph{(UID)} that each node possesses. Each node $u$ starts with a UID that is drawn from a namespace  $\mathcal{U}$. The maximum id is represented by $O(\log n)$ bits. An algorithm is called \emph{comparison based} if it manipulates the UIDs of the network using comparison operations $(<,>,=)$ only. All of the algorithms and lower bounds presented in this paper are comparison based. 

The nodes represent agents equipped with computation, communication, and edge-modification capabilities and they operating in synchronous rounds. In each round all agents perform the following actions in sequence and in lock step: 
Send messages to their neighbors, Receive messages from their neighbors, Activate edges with potential neighbors, Deactivate edges with neighbors, Update their local state.

We note that a node may choose to send a different message to different neighbors in a round and that the time needed for internal computations is assumed throughout to be $O(1)$. We do not impose any restriction on the size of the local memory of the agents, still the space complexity of our algorithms is within a reasonable polynomial in $n$.

\subsection{Problem Definitions and Performance Measures}

For the current paper we are mainly interested in the following problems.

\smallskip

\textbf{Leader Election.} Every node $u$ in graph $D=(V,E)$ has a variable $status_u$ that can be set to a value in $\{$Follower, Leader$\}$. An algorithm $A$ solves leader election if the algorithm has terminated and exactly one node has its status set to Leader while all other nodes have their status set to Follower.

\textbf{Token Dissemination.} Given an initial graph $D=(V,E)$ where each node $u\in V$ starts with some unique piece of information (token), every node $u\in V$ must terminate while having received that unique piece of information from every other node $v\in V\setminus \{u\}$. W.l.o.g. we will consider that unique information to be the UID of each node throughout the paper.

\textbf{Depth-$d$ Tree.} Given any initial graph $G_s$ from a given family, the distributed algorithm must reconfigure the graph into a target graph $G_f$, such that $G_f$ is a rooted tree of depth $d$ with a unique leader elected at the root.

\smallskip

Apart from studying the running time of our algorithms, measured as their worst-case number of rounds to carry out a given task, we also introduce the following measures related to their edge complexity.

\smallskip

\textbf{Total Edge Activations.} The total number of edge activations of an algorithm is given by $\sum_{i=1}^{T} |E_{ac}(i)|$, where $T$ is the running time of the algorithm. 

\textbf{Maximum Activated Edges.} It is defined as $\max_{i\in [T]} |E(i)\setminus E(1)|$, that is, equal to the maximum number of active edges of a round, disregarding the edges of the initial network.

\textbf{Maximum Activated Degree.} The maximum degree of a round, if we again only consider the edges that have been activated by the algorithm. Let $deg(G)$ denote the degree of a graph $G$. Then, formally, the maximum activated degree is equal to $\max_{i\in [T]} deg(D(i)\setminus D(1))$, where the graph difference is defined through the difference of their edge sets.

In this paper, instead of measuring the maximum activated degree we will focus on preserving the maximum degree of input networks from specific families. For example, one of our algorithms solves the Depth-$d$ Tree problem on any input network and, if the input network has bounded degree, then it guarantees that the degree in any round is also bounded.

\subsection{Basic Subroutines}
\label{SUBSEC:BASIC}

We will now provide algorithms that transform initial graphs into graphs with small diameter and which will be used as subroutines in our general algorithms. 

The first called TreeToStar transforms any initial rooted tree graph into a spanning star in $O(\log n)$ time with $O(n\log n)$ total edge activations and $O(n)$ active edges per round, provided that the nodes have a sense of orientation on the tree. This means that each node can distinguish its parent from its children. In every round, each node activates an edge with the potential neighbor that is its grandparent and deactivates the edge with its parent. This process keeps being repeated by each node until they activate an edge with the root of the tree.

\begin{proposition} \label{TreeToStar}
	Let $T$ be any tree rooted at $u_0$ of depth $d$. If the nodes have a sense of orientation on the tree, then algorithm TreeToStar transforms T into a spanning star centered at $u_0$ in $\ceil*{\log d}\leq\log n$ rounds. TreeToStar has at most $(n-1)+(n-2)=2n-3$ active edges per round. 
\end{proposition}

Our next algorithm called LineToCompleteBinaryTree transforms any line into a binary tree in $O(\log n)$ time with $O(n\log n)$ total edge activations, with $O(n)$ active edges per round and the degree of each node is at most $4$, provided that the nodes have a common sense of orientation. 

\noindent
In each round, each node activates an edge with its grandparent and afterwards it deactivates its edge with its parent. This process keeps being repeated by each node until they activate an edge with the root of the tree or if their grandparent has $2$ children.

\begin{proposition}
	Let $T$ be any line rooted at $u_0$ of diameter $d$. If the nodes have a sense of orientation on the line, then algorithm LineToCompleteBinaryTree transforms T into a binary tree centered at $u_0$ in $\ceil {\log d}\leq\log n$ time. LineToCompleteBinaryTree has at most $(n-1)+(n-2)=2n-3$ active edges per round, $n\log n$ total edge activations and bounded degree equal to $3$. 
\end{proposition}

\subsection{General Strategy for Depth-$d$ Tree}

All algorithms developed in this paper solve the Depth-$d$ Tree problem starting from any initial network $G_s$ from a given family. Our aim is to always achieve this in (poly)logarithmic time while minimizing some of the edge-complexity parameters. There is a natural trade-off between time and edge complexity and each of our algorithms makes a different contribution to this trade-off. In particular, by paying for linear degree, our first algorithm manages to be optimal in all other parameters. If we instead insist on bounded degree, then our second algorithm shows that we can still solve Depth-$d$ Tree within an additional $O(\log n)$ factor both in time and total edge activations. Finally, if the bound on the degree is slightly relaxed to (poly)log$(n)$ then our third algorithm achieves $o(\log^2 n)$ time. 

All three algorithms are built upon the same general strategy that we now describe. For each of them we choose an appropriate gadget network, which has the properties of being ``close'' to the target network $G_f$ to be constructed and of facilitating efficient growth. For example, the $G_f$ of our first algorithm is a spanning star and the chosen gadget is a star graph, while the $G_f$ of our second algorithm is a complete binary tree and the chosen gadget is the union of a ring and a complete binary tree spanning that ring (called a \emph{wreath}).

Our algorithms satisfy the following properties. The nodes are always partitioned into committees, where each committee is internally organized according to the corresponding gadget network of the algorithm and has a unique leader, which is the node with maximum id in that committee. Initially, every node forms its own trivial committee and committees increase their size by competing with nearby committees. In particular, committees select and, if possible, merge with the maximum-id committee in their neighborhood. Prior to merging, such selections may give rise to pairs of committees, in which case merging is immediate, but also to rooted trees of committees where all selections are oriented towards the root and merging has to be deferred. In the latter case, the winning committee will eventually be the root of the tree, at which point all other committees of the tree will have merged to it. In all cases, merging must be done in such a way that the gadget-like internal structure of the winning committee is preserved. This growth guarantees that eventually there will be a single committee spanning the network. At that point, the leader of that committee (which is always the node with maximum id in the network) is an elected unique leader. Moreover, the gadget-like internal structure of that committee can be quickly transformed into the desired target network, due to the by-design close distance between them. For example, in the algorithm forming a star no further modification is required, while in the algorithm forming a complete binary tree, a ring is eliminated from a wreath so that only the tree remains.  

Our algorithms are designed to operate in asynchronous phases, with the guarantee that in every phase pairs of committees merge and trees of committees halve their depth. This can be used to show that in all our algorithms a single committee will remain within $O(\log n)$ phases. Each phase lasts a number of rounds which is within a constant factor of the maximum diameter of a committee involved in it, which is in turn upper bounded by the diameter of the final spanning committee. The latter is always equal to the diameter of the chosen gadget as a function of its size. The total time is then given by the product of the number of phases and the diameter of the chosen gadget. For example, in our first algorithm the gadget is a star and the running time (in rounds) is $O(1)\cdot O(\log n)$, in our second algorithm the gadget is a wreath of diameter $O(\log n)$ and the running time is $O(\log n)\cdot O(\log n)=O(\log^2 n)$, while in our third algorithm the gadget is a modified wreath, called ThinWreath, of diameter $o(\log n)$ and the running time is $o(\log n)\cdot O(\log n)=o(\log^2 n)$. Given that every node activates at most one edge per round, the total number of edge activations of our algorithms is within a linear factor of their running time.  

\section{An Edge Optimal Algorithm for General Graphs}
\label{sec:algorithm-1}

Our first algorithm, called GraphToStar, solves the Depth-$d$ Tree problem, for $d=1$. In particular, by using a star gadget it transforms any initial graph $G_s$ into a target spanning star graph $G_f$. Its running time is $O(\log n)$ and it uses an optimal number of $O(n\log n)$ total edge activations and $O(n)$ active edges per round. Optimality is established by matching lower bounds, presented in Section \ref{SEC:LOWER-BOUNDS}.

\smallskip

\noindent\textbf{Algorithm GraphToStar} 

\smallskip

Each committee $C(u)$ is a star graph where the center node $u$ is the leader of the committee and all other nodes are followers. The leader node of each committee is the node with the greatest UID in that committee. The UID of each committee is defined by the UID of that committee's leader. The winning committee in the final graph, denoted $C(u_{max})$, is the one with the greatest UID in the initial graph. Every node starts as a leader and forms its own committee as a single node. 
The original edges of $G_s$ are assumed to be maintained until the last round of the algorithm and the nodes can always distinguish them. The algorithm proceeds in phases, where in every phase each committee $C(u)$ executes in one of the following modes, always executing in selection mode in phase 1.

\begin{itemize}

\item \textbf{Selection:} If $C(u)$ has a neighboring committee $C(z)$ such that $UID_z>UID_u$ and $C(z)$ is not in \emph{pulling mode}, then, from its neighboring committees not in pulling mode, $C(u)$ \emph{selects} the one with the greatest UID; call the latter $C(v)$. It does this, by $u$ first activating an edge $e_1$ with a potential neighbor in $C(v)$. Then $u$ activates an edge with $v$, deactivates the previous edge $e_1$, and $C(u)$ enters either the merging or pulling mode. In particular, if $C(v)$ did not select, then $C(u)$ and $C(v)$ form a pair and $C(u)$ enters the merging mode. If on the other hand $C(v)$ selected some $C(w)$, then $C(u)$ enters the pulling mode. 

Otherwise, $C(u)$ did not select. If $C(u)$ was selected then it enters the waiting mode, else it remains in the selection mode.

If $C(u)$ has no neighboring committees, then it enters the termination mode. 

\item \textbf{Merging:} Given that in the previous phase the leader of $C(u)$ activated an edge with the leader of $C(v)$, each follower $x$ in $C(u)$ activates the edge $xv$ and deactivates the edge $xu$. The result is that $C(u)$ and $C(v)$ have merged into committee $C(v)$, which remains a star rooted at $v$ now spanning all nodes in $V(C(u))\cup V(C(v))$. Therefore, $C(u)$ does not exist any more.

\item \textbf{Pulling:} Given that in the previous phase the leader of $C(u)$ activated an edge with the leader of $C(v)$ and the leader of $C(v)$ activated an edge with the leader of $C(w)$, $u$ activates $uw$, deactivates $uv$, and $C(u)$ remains in pulling mode. If, instead, the leader of $C(v)$ did not activate in the previous phase, then $C(u)$ enters the merging mode.

\item \textbf{Waiting:} If $C(u)$ has no neighboring committees, $C(u)$ enters the termination mode. If in the previous phase no committee $C(v)$ activated an edge with $u$, then $C(u)$ enters the selection mode. Otherwise $C(u)$ remains in the waiting mode.

\item \textbf{Termination:} $C(u)$ deactivates every edge in $E(G_s)\setminus E(C(u))$. In particular, each follower $x$ in $C(u)$ deactivates all active edges incident to it but $xu$. 
\end{itemize}

\smallskip

\noindent\textbf{Correctness}

	\begin{lemma}
	Algorithm GraphToStar solves Depth-1 Tree.
	\end{lemma}

	\begin{proof}
It suffices to prove that in any execution of the algorithm, one committee eventually enters the termination mode and that this committee can only be $C({u_{max}})$. If this holds, then by the end of the termination phase $C({u_{max}})$ forms a spanning star rooted at $u_{max}$ and $u_{max}$ is the unique leader of the network. This satisfies all requirements of Depth-1 Tree.

A committee \emph{dies} (stops existing) only when it merges with another committee by entering the merging mode. First observe that there is always at least one \emph{alive} committee. This is $C({u_{max}})$, because entering the merging mode would contradict maximality of $u_{max}$. We will prove that any other committee eventually dies or grows, which due to the finiteness of $n$ will imply that eventually $C({u_{max}})$ will be the only alive committee.

In any phase, but the last one which is a termination phase, it holds that every alive committee $C(u)$ is in one of the selection, merging, pulling, and waiting modes. If $C(u)$ is in the merging mode, then by the end of the current phase it will have died by merging with another committee $C(v)$. It, thus, remains to argue about committees in the selection, pulling, and waiting modes. 

We first argue about committees in the pulling mode. Denote their set by $\mathcal{C}_p$. Observe that, in any given phase, the committees in pulling mode form a forest $F$, where each $C(u)\in \mathcal{C}_p$ belongs to a tree $T$ of $F$. Any such tree executes the TreeToStar algorithm (from Section \ref{SUBSEC:BASIC}) on committees and satisfies the invariant that its root committee $C_r$ is always in the waiting mode and $C_r$'s children are in the merging mode. In every phase, $C_r$'s children merge with $C_r$ and their children become the new children of $C_r$ and enter the merging mode. It follows that all non-root committees in $T$ will eventually merge with $C_r$. Thus, all committees in pulling mode eventually die.

It remains to argue about committees in the selection and waiting modes. We start from the waiting mode. Any committee $C(u)$ in waiting mode is a root of either a tree in the forest $F$ 
or of a star of committees in which all leaf-committees are merging with $C(u)$. In both cases, $C(u)$ eventually exits the waiting mode and enters the selection mode. This happens as soon as all other committees in its tree or star have merged to it, thus $C(u)$ has grown upon its exit.    

Now, a committee $C(u)$ in the selection mode can enter any other mode. As argued above, if it enters the merging or pulling modes it will eventually die and if it enters the waiting mode it will eventually grow. Thus, it suffices to consider the case in which it remains in the selection mode indefinitely. This can only happen if all current and future neighboring committees of $C(u)$, including the ones to eventually replace neighbors in pulling mode, have an id smaller than $UID_u$. But each of these must have selected a neighboring $C(w)$, such that $UID_w>UID_u$, otherwise it would have selected $C(u)$. Any such selection, results in $C(w)$ (or a $z$, such that $UID_z>UID_w$ in case $w$ belongs to a tree) becoming a neighbor of $C(u)$, thus contradicting the indefinite local maximality of $UID_u$.
	
	\end{proof}

	\noindent\textbf{Time Complexity}

	Let us move on to proving the time complexity of our algorithm. At the beginning, we are going to ignore the number of rounds within a phase, and we are just going to study the maximum number of phases before a single committee is left. We define $S(C(u_s))$ to be the \emph{size} of committee $C(u)$ in phase $s$.

	\begin{lemma}\label{trees}
		Consider committee $C(v)$ that is in waiting mode between phases $s$ and $s+j$. If the size of every committee in phase $s$ is at least $2^k$, then the size of committee $C(v)$ once it enters the selection mode in phase $s+j+1$ is at least $2^{k+j}$.
	\end{lemma}
	
	\begin{proof}
	Any committee $C(u)$ in waiting mode is a root of (i) either a tree in the forest $F$ 
	or (ii) a star of committees in which all leaf-committees are merging with $C(u)$. 
	
	For case (i): root committee $C(u)$ is always in waiting mode and $C(u)$'s children are in merging mode. In every phase, $C(u)$'s children merge with $C(u)$ and their children become the new children of $C(u)$ and enter the merging mode. It follows that all non-root committees in the tree will eventually merge with $C(u)$ in some phase $j$. Note that due to the nature of the pulling mode, in each phase the children of $C(u)$ are doubled. This is true because the pulling mode is simulating the TreeToStar algorithm on committees. Recall that we assumed that the size of every committee is $S(C(v_s))\geq2^k$ in phase $s$. Then in each phase $s+i$, where $0<i\leq j$, the size of the root committee is $S(C(u_{s+\log i}))= S(C(u_s))+2\cdot S(C(v_s))+4\cdot S(C(v_s))+\ldots+2^{\log (i-1)}\cdot S(C(v_s))=2^{k+i}$.
	
	For case (ii): root committee $C(u_s)$ is in waiting mode and has at least one leaf committee in phase $s$. After the leaf committee merges in $1$ phase, committee $C(u_{s+1})$ has size $S(C(u_{s+1}))\geq S(C(u_s))+S(C(u_s))=2^k+2^k=2^{k+1}.$
	\end{proof}

	\begin{lemma}\label{wait}
	If committee $C(u)$ stays in the selection mode for $p\geq 4$ consecutive phases, then $C(u)$ has a neighboring committee $C(v)\in \mathcal{C}_p$ that belongs to a tree $T$ for at least $p$ phases.
	\end{lemma}	

	\begin{proof}
	Let us assume that committee $C(u)$ stays in the selection mode for $p\geq 4$ consecutive phases while having a neighbor $C(v)$ that does not belong to tree $T$. If $C(v)$ does not belong to a tree in phase $k$, then it cannot be in pulling mode. If $C(v)$ is in selection mode in phase $k$ and $C(v)$ does not select $C(u)$ and $C(u)$ does not select $C(v)$, then $C(v)$ has a neighbor $C(w)$ where $UID_w>UID_v>UID_u$ and $C(v)$ selected $C(w)$. Then $C(v)$ enters the merging mode in phase $k+1$ and gets merged with $C(w)$. In phase $k+2$ committee $C(w)$ becomes a neighbor of $C(v)$ and $C(w)$ enters the selection mode. Therefore $C(v)$ would select $C(w)$ in phase $k+2$ and exit the selection mode. Thus, a contradiction. If $C(v)$ is in waiting mode in phase $k$, it cannot be the root of a tree, and is the root of a star. Therefore in phase $k+1$ it will enter the selection mode and based on the analysis of the previous paragraph, in phase $k+3$ $C(u)$ will exit the selection mode. Thus, a contradiction.
	\end{proof}

	\begin{lemma}\label{minimum}
	Let us assume that the minimum size of a committee in phase $s$ is $2^k$. If committee $C(u)$ stays in the selection mode from phase $s$ to phase $s+p$ where $p\geq 4$ consecutive phases, then in phase $s+p+1$ it will select or get selected by a committee $C(v)$ of size $2^{k+p-4}$.
	\end{lemma}		
	
	\begin{proof}
	From Lemma \ref{wait} it follows that, since $C(u)$ is in the selection mode for at least $4$ phases, there exists a neighbor $C(v)$ that belongs to a tree $T$ with. Since $C(u)$ exits the selection mode in phase $s+p$, it either selects committee $C(w)$ that the root of tree $T$ or $C(w)$ selects $C(v)$. Since $C(u)$ was in the selection phase for $p$ phases, committee $C(w)$ was on a tree of depth at least $p-3$. From Lemma \ref{trees} it follows that the size of $C(w)$ is $2^{k+p-3}$.
	\end{proof}

	\begin{lemma}\label{final}
		Assume that the minimum size of every committee in phase $s$ is $2^k$ and that every committee will have exited the selection mode in phase $s+p$ at least once. The size of all winning committees in phase $p+1$ is at least $2^{k+p-4}$.
	\end{lemma}	
	
	\begin{proof}
	Trivially, if $p\le 4$ the winning committee has size at least $2^{k+1}$ in phase $p+1$ since it has merged with at least one other committee.
	
	From Lemma \ref{minimum} it follows that if $p\geq 4$ the winning committee between $C(w)$ and $C(u)$ will have size at least $2^{k+p-3}$ in phase $s+p+1$.
	\end{proof}

	\begin{lemma} 
		After $O(\log n)$ phases, there is only a single committee left in the graph.
	\end{lemma}
	
	\begin{proof}
	From Lemma \ref{final}, it follows that after $O(\log n)$ phases, there will be a committee with at least $2^{\log n}$ nodes.
	\end{proof}

	\begin{lemma}
	
	Each phase consists of at most $2$ rounds.
	\end{lemma}	

	\begin{proof}
	Based on the description of the modes, the selection phase lasts $2$ rounds, the pulling phase lasts $1$ round, the merging phase lasts $1$ round, the waiting phase lasts $1$ round and the termination phase lasts $2$ rounds.
	\end{proof}
	
	\noindent\textbf{Edge Complexity}
	
	It is very simple to prove the edge complexity for the algorithm. Note that in each round $i$ each node activates at most 1 edge. Furthermore, if a node had activated an edge $u$ in round $i$, and it activates another edge $v$ in round $i+1$, then it deactivates edge $u$. Therefore, each node cannot have more than $2$ active edges that it has activated itself at any time and since we have $n$ nodes in the network, there can ever be at most $2n$ active edges per round. 
	
	\begin{theorem}
		For any initial connected graph $G_s$, the \emph{GraphToStar} algorithm solves the Depth-1 Tree problem in $O(\log{n})$ time with at most $O(n\log n)$ total edge activations and $O(n)$ active edges per round.
	\end{theorem}

\section{Minimizing the Maximum Degree on General Graphs}
\label{SEC:ALGORITHM-2}

In the previous section, we devised an algorithm that minimizes the edge complexity of the graph but this came at a cost of linear degree. In this section we will create an algorithm that minimizes the maximum activated degree to a constant but has $O(\log^2 n)$ running time and $O(n\log^2 n)$ total edges activations. 

Recall the committees from section \ref{sec:algorithm-1}. Every committee was a star graph which was very practical. First of all, the leader of each committee $C(u)$ was a potential neighbor of each neighboring committee $C(v)$ and therefore $u$ could communicate in $O(1)$ phases with every $C(v)$ and decide with which $C(v)$ to merge with. Additionally merging committee $C(v)$ with $C(u)$ required $O(1)$ phases. Finally, the pulling phase cannot be used to merge multiple committees fast in this section, since it does not guarantee a constant degree for every node All of the above techniques were possible due to the small diameter of the star and the linear degree of each node. 

For this algorithm, our committees must have at least $\Omega (\log n)$ diameter in order to have a constant degree and therefore merging two different committees in constant time while keeping a specific structure proves to be complicated. The new gadget of our committees is going to be a graph we call \emph{wreath}. A wreath graph is a graph that has both a ring subgraph and a complete binary tree subgraph. We are going to use the edges of the ring subgraph to merge committees and the binary tree subgraph to exchange information between the nodes of the graph. First, let us define the structure of the wreath graph.

\begin{definition}
	We define a graph $D=(V,E)$ to belong to the wreath class of graphs if it has two subgraphs $D_r=(V,E_r)$ and $D_b=(V,E_b)$, where $D_r=(V,E_r)$ belongs to the class of ring graphs, $D_b=(V,E_b)$ belongs to the class of complete binary tree graphs, and $E=E_r\cup E_b$. 
\end{definition}

The $O(\log n)$ diameter that the wreath graph possesses, will allow the leaders of committees $C(u)$ to communicate with neighboring committees $C(v)$ in $O(\log n)$ time. Additionally, the merging phase of each pair of committees will require only $O(\log n)$ time.  The algorithm is almost identical to the GraphToStar as far as the high level strategy is concerned. Committees select neighboring committees and merge with them. The main difference is that when a tree with root $w$ is formed, we cannot use the pulling mode since this would increase the degree significantly. Instead the committees on each tree merge in a single ring that includes all committees in $O(1)$ time (ring merging mode). After this, $w$ deactivates one of its incident edges in order to create a line subgraph. Once this happens, each node on the ring executes an asynchronous version of the LineToCompleteBinaryTree subroutine in $O(\log n)$ time using the orientation of the new ring, where root $w$ is the root of the line. Once the subroutine is finished, the complete binary tree subgraph of the wreath graph is ready. Therefore we have managed to merge a tree graph of multiple committees into a single committee. 

\smallskip

\noindent\textbf{Algorithm GraphToWreath} 

\smallskip
Each committee $C(u)$ is a wreath graph where  $u$ is the leader of the committee and all other nodes are followers. The leader node of each committee is the node with the greatest UID in that committee. The UID of each committee is defined by the UID of that committee's leader. The winning committee in the final graph is the one with the greatest UID in the initial graph. Every node starts as a leader and forms its own committee as a single node. We will sometimes refer to a committee by its leader's name. The original edges of $G_s$ are assumed to be maintained until the last round of the algorithm and the nodes can always distinguish them. Our algorithm proceeds in phases, where in every phase each committee $C(u)$ executes in one of the following modes, always executing in selection mode in phase 1.

\begin{itemize}
	
	\item \textbf{Selection:} If $C(u)$ has a neighboring committee $C(z)$ such that $UID_z>UID_u$ and $C(z)$ is not in \emph{Ring Merging mode} or \emph{Tree Merging mode} then, from its neighboring committees not in ring merging or tree merging mode, $C(u)$ \emph{selects} the one with the greatest UID; call the latter $C(v)$. If $C(u)$ selected $C(v)$ or $C(u)$ was selected, $C(u)$ enters the Ring Merging mode. If $C(u)$ did not select anyone and $C(u)$ was not selected by anyone, it stays in the selection mode. If $C(u)$ has no neighboring committees, $C(u)$ enters the termination mode.

	\item \textbf{Ring Merging:} Given that in the previous phase, $C(u)$ selected $C(v)$, committee $C(u)$ merges its ring component with the ring component of $C(v)$ as described in the Merging the Spanning Ring Subgraph paragraph (see Appendix). Given that in the previous phase, $C(u)$ was selected by $C(k)$, committee $C(k)$ merges its ring component with the ring component of $C(u)$ as described in the Merging the Spanning Ring Subgraph paragraph. $C(u)$ enters the tree merging mode.

	\item \textbf{Tree Merging:} Every node $x$ in $C(u)$ executes one round of the asynchronous LineToCompleteBinaryTree algorithm as described in the asynchronous LineToCompleteBinaryTree paragraph (see Appendix). If there exists node $x$ that has not terminated the asynchronous LineToCompleteBinaryTree algorithm, $C(u)$ stays in the Tree Merging mode. If all nodes $x$ have terminated the asynchronous LineToCompleteBinaryTree algorithm, all nodes $x$ have now merged with committee $C'(u)$ whose leader is the root of the complete binary tree and $C'(u)$ enters the selection mode. $C(u)$ does not exist anymore.
	
	\item \textbf{Termination:} Each follower $x$ in $C(u)$ deactivates every edge apart from the edges that define the spanning complete binary tree subgraph.
\end{itemize}

Note here that we omit the communication steps for clarity and we claim that any communication performed between neighboring committees can be completed in $O(\log n)$ rounds since the diameter of each committee is at most $O(\log n)$.	

\begin{theorem}
	For any initial connected graph with constant degree, the GraphToWreath algorithm solves Depth-$\log n$ Tree in $O(\log^2 n)$ time with $O(n\log^2 n)$ total edge activations, $O(n)$ active edges per round and $O(1)$ maximum activated degree.
\end{theorem}

\section{Trading the Degree for Time}
\label{SEC:ALGORITHM-3}

For our new algorithm, we are going to try to have $O(\frac{\log n}{\log\log n})$ time for the merging but we are going to allow the maximum degree to reach $O(\log^2 n)$. This requires a new graph for our committees where the diameter of the shape is $O(\frac{\log n}{\log\log n})$, so that the communication within the committees is $O(\frac{\log n}{\log\log n})$ and a new way to merge the committees in $O(\frac{\log n}{\log\log n})$. We also have to make the assumption that all nodes know the size of the initial graph.

Our new graph is very similar to the Wreath graph and we call it \emph{ThinWreath}. The main difference is that instead of having a complete binary tree component, it has a complete polylogarithmic degree tree component with diameter $O(\frac{\log n}{\log\log n})$ and polylogarithmic degree. The $O(\frac{\log n}{\log\log n})$ diameter that the ThinWreath graph possesses, will allow the leaders of committees $C(u)$ to communicate with neighboring committees $C(v)$ in $O(\frac{\log n}{\log\log n})$ time. 

\smallskip

\noindent\textbf{Algorithm GraphToThinWreath} 

\smallskip

Each committee $C(u)$ is a ThinWreath graph where $u$ is the leader of the committee and all other nodes are followers. The leader node of each committee is the node with the greatest UID in that committee. The UID of each committee is defined by the UID of that committee's leader. The winning committee in the final graph is the one with the greatest UID in the initial graph. Every node starts as a leader and forms its own committee as a single node. We will sometimes refer to a committee by its leader's name. The original edges of $G_s$ are assumed to be maintained until the last round of the algorithm and the nodes can always distinguish them. We also have to assume that the nodes know the size of the initial graph. Our algorithm proceeds in phases, where in every phase each committee $C(u)$ executes in one of the following modes, always executing in selection mode in phase 1.

\begin{itemize}
	
	\item \textbf{Selection:} If $C(u)$ has a neighboring committee $C(z)$ such that $UID_z>UID_u$ and $C(z)$ is in selection mode, then, from its neighboring committees in the selection mode, $C(u)$ \emph{selects} the one with the greatest UID; call the latter $C(v)$. If $C(u)$ was selected, $C(u)$ enters the \emph{Matchmaker mode}. If $C(u)$ was not selected and $C(u)$ selected $C(v)$, $C(u)$ enters the \emph{Matched mode}. If $C(u)$ did not select anyone and $C(u)$ was not selected by anyone, it stays in the selection mode. If $C(u)$ has no neighboring committees, $C(u)$ enters the termination mode.

	\item \textbf{Matchmaker:} If committees $C(k)$ had selected $C(u)$ in the previous phase, committee $C(u)$ matches committees $C(k)$ in pairs. If the number of committees $C(k)$ that selected $C(u)$ is odd, one committee is matched with $C(u)$. $C(u)$ enters the Matched mode.
	
	\item \textbf{Matched:} If committee $C(u)$ selected committee $C(v)$ in the last selection phase, committee $C(u)$ learns with which committee it has been matched. Committee $C(u)$ enters the Ring Merging mode.
	
	\item \textbf{Ring Merging:} Given that in the previous phase, $C(u)$ was matched with $C(v)$, committee $C(u)$ merges its ring component with the ring component of $C(v)$ as described in the Merging the Spanning Ring Subgraph paragraph (see Appendix) where the winning committee is $C(u)$ if $UID_u>UID_v$ and vice versa. Committee  $C(u)$ enters the Leader Merging mode.

	\item \textbf{Leader Merging:} Given that in the previous phase committee $C(u)$ lost to committee $C(k)$, the leader of $C(u)$ activates an edge with the leader of $C(k)$. 
	If committee $C(k)$ has lost to some other committee $C(l)$ in the previous phase, $C(u)$ enters the Tree Merging mode. If $C(u)$ did not lose to any other committee, $C(u)$ enters the Tree Merging mode where $u$ is the root.
	
	\item \textbf{Tree Merging:} The leader of $C(u)$ executes one round of the asynchronous LineToCompletePolylogarithmicTree algorithm as described in the Asynchronous LineToCompletePolylogarithmicTree paragraph (see Appendix). If there exists node $x$ that has not terminated the asynchronous LineToCompletePolylogarithmicTree algorithm, $C(u)$ stays in the Tree Merging mode. If all nodes $x$ have terminated the asynchronous LineToCompletePolylogarithmicTree algorithm, all nodes $x$ have now merged with committee $C'(u)$ whose leader $u'$ is the root of the complete polylogarithmic tree and $C'(u)$ enters the selection mode. Committee $C(u)$ does not exist anymore.
	
	\item \textbf{Termination:} Each follower $x$ in $C(u)$ deactivates every edge apart from the edges that define the spanning complete polylogarithmic tree subgraph.
\end{itemize}

\begin{theorem}
	For any initial connected graph with polylogarithmic degree, the GraphToThinWreath algorithm solves Depth-$\log n$ Tree in $O(\frac{\log^2 n}{\log\log n})$ time with $O(n\log^2 n)$ total edge activations, $O(n)$ active edges per round and $O(1)$ maximum activated degree. 
\end{theorem}

\section{Lower Bounds for the Depth-$\log n$ Tree Problem}
\label{SEC:LOWER-BOUNDS}

We now shift our focus into proving lower bounds for Depth-$\log n$ Tree. 

\begin{lemma} \label{LowerTimeBound}
	Any centralized transformation strategy requires $\Omega(\log n)$ rounds to solve Depth-$\log n$ Tree if the initial graph $G_s$ is a spanning line.
\end{lemma}

\begin{lemma}\label{EdgeActivations}
	Any centralized transformation strategy that solves Depth-$\log n$ Tree in $O(\log n)$ rounds, requires $\Omega(n)$ edge activations and $\Omega (n/\log n)$ edge activations per round.
\end{lemma}

On the positive side:

\begin{theorem}
	There is a centralized transformation strategy that, for any initial graph $D=(V,E)$, solves Depth-$\log n$ Tree in $O(\log n)$ rounds, with $\Theta(n)$ total edge activations.
\end{theorem}

We are now going to show that there is a difference in the minimum total edge activations required for solving the Depth-$\log n$ Tree problem between the centralized and the distributed case.

\begin{theorem}
	Any distributed algorithm that solves the Depth-$\log n$ Tree problem in $O(\log n)$ time, requires $\Omega (n\log n)$ total edge activations.
\end{theorem}

\section{Conclusion and Open Problems}

In this work we considered a distributed model for actively dynamic networks. The model can achieve global distributed computation and network reconfiguration in (poly)logarithmic time, but trivial solutions incur an impractical cost, which is related to the creation and maintenance of edges in the dynamic network generated by the algorithm. We defined natural cost measures associated with the edge complexity of actively dynamic algorithms. It turns out that there is a natural trade-off between the time and edge complexity of algorithms. By focusing on the apparently representative task of transforming any initial network from a given family into a target network of (poly)logarithmic diameter, which can then be exploited for global computation or further reconfiguration, we obtained non-trivial insight into this trade-off.

Our model is inspired by recent developments in the algorithmic theory of dynamic networks and in the theory of reconfigurable robotics. Still, it turns out to be very close to the interesting area of overlay network construction. It is not clear yet what is the formal relationship between the polylogarithmic restriction on communication in overlay networks and our efforts to minimize the total number of edge activations in our algorithms. This remains an interesting question for future research.

There is also a number of technical questions specific to our model and the obtained results. We do not know yet what are the ultimate lower bounds on time for different restrictions on the maximum degree. For maximum degree bounded by a constant our best upper bound is $O(\log^2 n)$ and if bounded by (poly)log$(n)$ this drops slightly by an $O(\log\log n)$ factor. Can any of these be improved to $O(\log n)$, that is, matching the $\Omega(\log n)$ lower bound on time? It would also be valuable to investigate randomized algorithms for the same problems, like the ones already developed in overlay networks.

Finally, there are many variants of the proposed model and complexity measures that would make sense and might give rise into further interesting questions and developments. Such variants include anonymous distributed entities which are possibly restricted to treat their neighbors identically even w.r.t. actions (e.g., through local broadcast) and alternative potential neighborhoods, e.g., activating edges at larger distances.

\bibliographystyle{ACM-Reference-Format}
\bibliography{MSS20-arxiv}

\newpage

\noindent\textbf{\Large APPENDIX}
\appendix

\section{Omitted Details from Section \ref{SUBSEC:BASIC} - Basic Subroutines}
\label{appSUBSEC:BASIC}

We will now provide algorithms that transform initial graphs into graphs with small diameter and which will be used as subroutines in our general algorithms. 

The first called TreeToStar transforms any initial rooted tree graph into a spanning star in $O(\log n)$ time with $O(n\log n)$ total edge activations and $O(n)$ active edges per round, provided that the nodes have a sense of orientation on the tree. This means that each node can distinguish its parent from its children. In every round, each node activates an edge with the potential neighbor that is its grandparent and deactivates the edge with its parent. This process keeps being repeated by each node until they activate an edge with the root of the tree.

\begin{proposition} \label{appTreeToStar}
	Let $T$ be any tree rooted at $u_0$ of depth $d$. If the nodes have a sense of orientation on the tree, then algorithm TreeToStar transforms T into a spanning star centered at $u_0$ in $\ceil*{\log d}\leq\log n$ rounds. TreeToStar has at most $(n-1)+(n-2)=2n-3$ active edges per round. 
\end{proposition}

\begin{proof}
	Just before deactivating edges in the current round the set of active edges consists of the $(n-1)$ edges of the tree in the previous round plus at most one new parent connection per node in all but the top 2 levels of the tree. As the top 2 have at least 2 nodes, there are at most $n-2$ edge activations per round. Then edges are being deleted resulting in a tree by the end of each phase therefore the bound holds for all rounds.
	
	Recall that the algorithm runs for $\ceil*{\log d}<\log n$ rounds and there are at most $n-2$ edge activations per round. Therefore we have $O(\log n)$ total edge activations. 	
\end{proof}	

Our next algorithm called LineToCompleteBinaryTree transforms any line into a binary tree in $O(\log n)$ time with $O(n\log n)$ total edge activations, with $O(n)$ active edges per round and the degree of each node is at most $4$, provided that the nodes have a common sense of orientation. This means that each node can distinguish its parent from its children. 

\noindent
In each round, each node activates an edge with its grandparent and afterwards it deactivates its edge with its parent. This process keeps being repeated by each node until they activate an edge with the root of the tree or if their grandparent has $2$ children.

\begin{proposition}
	Let $T$ be any line rooted at $u_0$ of diameter $d$. If the nodes have a sense of orientation on the line, then algorithm LineToCompleteBinaryTree transforms T into a binary tree centered at $u_0$ in $\ceil {\log d}\leq\log n$. LineToCompleteBinaryTree has at most $(n-1)+(n-2)=2n-3$ active edges per round, $n\log n$ total edge activations and bounded degree equal to $3$. 
\end{proposition}

\begin{proof}
	The proof for the total edge activations and active edges per round is identical to the proof of Proposition \ref{appTreeToStar} since the two algorithms have exactly the same execution but for the termination criteria. Therefore the edge performance analysis stays the same. For the bounded degree, by definition of the algorithm, in each round, each node $u$ activates an edge with its grandparent and deactivates an edge with its parent. Additionally, if node $u$ has two children, no other node $v$ may activate an edge with $u$. If we take into account the above statements, and the fact that each node starts with $2$ incident edges, then the maximum degree of each node throughout the execution of the algorithm is at most $4$.
\end{proof}

\section{Omitted Details from Section \ref{SEC:ALGORITHM-2} - Minimizing the Maximum Degree on General Graphs}

\smallskip

\noindent\textbf{Correctness} 

\smallskip

\begin{lemma}
	Algorithm GraphToStar solves $Depth-\log n$ Tree.
\end{lemma}

\begin{proof}
	It suffices to prove that in any execution of the algorithm, one committee eventually enters the termination mode and that this committee can only be $C({u_{max}})$. If this holds, then by the end of the termination phase $C({u_{max}})$ forms a spanning complete binary tree rooted at $u_{max}$ and $u_{max}$ is the unique leader of the network. This satisfies all requirements of Depth-$\log n$ Tree.
	
	A committee \emph{dies} only when it merges with another committee by entering the tree merging mode. First observe that there is always at least one \emph{alive} committee. This is $C({u_{max}})$, because when it enters the tree merging mode, it is always the root of the complete binary tree. We will prove that any other committee eventually dies or grows, which due to the finiteness of $n$ will imply that eventually $C({u_{max}})$ will be the only alive committee.
	
	In any phase, but the last one which is a termination phase, it holds that every alive committee $C(u)$ is in one of the selection, ring merging, and tree merging modes. If $C(u)$ is in the ring merging mode then it will enter the tree merging mode and if its leader is not the root of the complete binary tree, then by the end of the current phase it will have died by merging with another committee $C'(u)$. It, thus, remains to argue about committees in the selection mode. 
	
	Now, a committee $C(u)$ in the selection mode can enter the tree merging mode. As argued above, if it enters the ring merging and tree merging modes in sequence it will either die or it will eventually grow. Thus, it suffices to consider the case in which it remains in the selection mode indefinitely. This can only happen if all current and future neighboring committees of $C(u)$ have an id smaller than $UID_u$. But each of these must have selected a neighboring $C(w)$, such that $UID_w>UID_u$, otherwise it would have selected $C(u)$. Any such selection, results in $C(w)$ becoming a neighbor of $C(u)$, thus contradicting the indefinite local maximality of $UID_u$.
\end{proof}

\noindent\textbf{Time Complexity}

\smallskip 

Let us move on to proving the time complexity of our algorithm. At the beginning, we are going to ignore the number of rounds within a phase, and we are just going to study the maximum number of phases before a single committee is left.

\begin{lemma} 
	After $O(\log n)$ phases, there is only a single committee left in the graph.
\end{lemma}

\begin{proof}

	Note that there is a direct correspondence between the modes in the GraphToWreath algorithm and the GraphToStar algorithm. 
	
	Both selection modes are used to decide the selections between the neighboring committees. The difference between the two algorithms is that each selection phase has a different running time. In particular, The GraphToStar selection phase required $2$ rounds while the selection phase of the GraphToWreath requires $O(\log n)$ rounds due to the diameter of the Wreath graph that each committee has.
	
	The ring mode is always an intermediate phase between the selection phase and the tree merging phase that lasts for $O(1)$ rounds. The purpose of this mode is to turn the tree $T$ created by the committees in the selection phase into a line so that the LineToCompleteBinaryTree subroutine can work.
	
	The pulling mode in the GraphToStar implements the TreeToStar subroutine, while the tree merging mode in the GraphToWreath implements the asynchronous version of the LinetoCompleteBinaryTree. Both subroutines are used to merge the Trees $T$ of depth $t$ created by the committees in $O(\log t)$ time and recall from the basic subroutines subsection that the TreeToStar and the LineToCompleteBinaryTree have the same running time. Therefore both algorithms require the same amount of phases.
	
	Note that there is no merging or waiting mode in the GraphToWreath since those modes have also been implemented by the merging tree mode.
	
	Since all modes that have been implemented in the GraphToWreath have equivalent modes in the GraphToStar with similar running times, the GraphToWreath algorithm requires at most $O(\log n)$ phases.
\end{proof}	

\begin{lemma}
	Each phase in the GraphToWreath algorithm, requires at most $O(\log n)$ rounds.
\end{lemma}

\begin{proof}
	We defer the detailed analysis for the proof to the low level description of the GraphToWreath algorithm where each different mode is described in detail with the upper bound on its running time.
	
	In short, let us argue that the selection phase requires $O(\log n)$ rounds since each $C(u)$ has to exchange information with its neighboring committees in order to decide with which committee $C(w)$ it is gonna merge with and whether any other committee $C(v)$ will decide to merge with $C(u)$. This requires time that is upper bounded by the diameter of each committee.
	
	The ring merging phase requires $O(1)$ rounds since every committee has to merge its ring component with committees $C(v)$ and the running time does not depend on the size of each committee participating.
	
	The tree merging mode implements one round of the asynchronous LineToCompleteBinaryTree.
	
	Therefore each phase requires at most $O(\log n)$ rounds to execute.
\end{proof}

\noindent\textbf{Edge Complexity}

\smallskip

The analysis for the total edge activations is simple. The algorithm runs for $O(\log^2 n)$ rounds and each node activates at most $1$ edge per round. Therefore the total edge activations are $O(n\log^2 n)$. 

Let us consider the maximum incident edges that a node can have, excluding the edges of the initial graph. Based on the low level description of the GraphToWreath algorithm, a node can have $2$ active incident edges for its ring neighbors in the wreath graph and another $2$ when two ring graphs are merging. Additionally, it can have $2$ active edges for the complete binary tree and another $2$ for the execution of the LineToCompleteBinaryTree. Therefore the active edges per round are $O(n)$ and the maximum degree of each node is $8+c$, where $c$ is the degree of each node in the original graph.

\subsection{Merging Wreath Graphs}

We next give the low level description of the GraphToWreath algorithm.

Consider multiple committees two of which are $C(u)=(K,L(p),u)$ and $C'(u)=(K',L'(p),v)$ where each committee is a wreath graph and $u,v$ are the leaders of each committee respectively. Every node in each committee knows the leader, the size and the diameter of the committee it belongs to and finally, all of the nodes share a common orientation based on the ring subgraph of the committee. W.l.o.g. we assume that the orientation is clockwise. Similar to the algorithm in section 4, each committee selects the neighboring committee with the greatest $UID$ to merge with. 

\smallskip

\noindent\textbf{Communication}

Here we describe how committees communicate with other committees and how they understand in which mode they are and in which mode they should switch to.

Each follower $x$ in committee $C(u)$ sends a message $\{myUID_x,maxNeighborUID,maxNeighborDiameter\}$ to its leader $u$, where $myUID_x$ contains the $UID$ of node $x$, $maxNeighborUID$ contains the $UID$ of the neighboring committee with the greatest UID among along neighboring committees that $x$ has an edge with, and $maxNeighborDiameter$ contains the diameter of that committee. This step requires at most $\log d\leq \log n$ rounds where $d$ is the diameter of committee $C(u)$.

After committee leader $u$ receives all the triplets $\{myUID_x,maxNeighborUID,maxNeihbourDiameter\}$, then if $\not\exists$ $maxNeighborUID>UID_u$, committee $C(u)$ does not select another committee and committee leader $u$ waits to see whether another committee has selected $C(u)$. Committee leader $u$ knows the maximum waiting time since it just received the maximum diameter of all neighboring committees. 
If $\exists maxNeighborUID>UID_u$, $C(u)$ selects the neighbouring committee $C(v)$ with the greatest $maxNeighborUID$ and sends a message to $x$ to initiate the connection with that committee. This step requires $\log d\leq \log n$ rounds.

After follower $x$ receives the initiation message, it sends a message to its neighbor $y\in K'$ that $C(u)$ has selected $C(v)$. Follower $y$ sends the selection request to its leader $v$. Leader $v$ waits enough in order to receive all possible selection requests from different neighboring committees. When $v$ receives all the requests, it sends back an approval message to all nodes $y$ with a timestamp that defines in which round the merging should happen. This step requires $2\log d+2\cdot maxNeighborDiameter\leq 4*\log n$ rounds.

Therefore every committee $C(u)$ can understand which committee $C(v)$ it has selected and whether any committees $C'(v)$ have selected $C(u)$. This means that $C(u)$ knows which mode it should enter after the selection phase.

\smallskip

\noindent\textbf{Merging the Spanning Ring Subgraph}

\smallskip

Suppose that multiple committees $C(u)$ have selected committee $C(v)$ and are trying to merge their ring component in the Ring Merging mode. For simplicity, let us call those committees $C(1),C(2),\ldots, C(n)$ and let us call $x_1\in V(C(1)), x_2\in V(C(2)),\ldots, x_p\in V(C(p))$ the nodes of each committee that are neighbors to $y\in C(v)$. The merging description follows. The clockwise neighbor of $y$ activates an edge with the counterclockwise neighbor of $x_1$, $x_1$ activates an edge with the counterclockwise neighbor of $x_2,\ldots$, and $x_{p-2}$ activates an edge with the counter clockwise neighbor of $x_{p-1}$. This process requires $2$ rounds because each pair of nodes has distance $3$ from each other. Afterwards, $y$ deactivates an edge with its clockwise neighbor and $x_1,x_2,\ldots,x_p$ deactivate the edge with their counterclockwise neighbor. Note here that at this point, $C(v)$ has a spanning ring that includes the nodes from all committees $C(1),C(2),\ldots, C(n)$. The above operations require $O(1)$ rounds since they do not depend on the size of each committee. After the operations are finished, every committee enters the tree merging mode.

Note here that the above process creates a graph with diameter $d\in O(\log n)$ if all committees involved were in a star subgraph, and a graph with diameter $d>O(\log n)$ if all committees involved were in a tree subgraph. Therefore we have to handle the two cases differently. Also note that there are some special cases where the above process needs to be tweaked in order to work e.g. when a committee consists of $1$ or $2$ nodes and it has not formed a ring yet.

\smallskip

\noindent\textbf{Merging the CompleteBinaryTree Subgraph}

\smallskip

\noindent\textbf{Stars}. Assume that committee $C(u)$ selected no committee but at least $1$ other committee $C(v)$ selected $C(u)$. After the ring merging is complete, committee leader $u$ sends a message to all nodes in committee $C(u)$ with a timestamp that defines in which round they should deactivate the edges of the previous complete binary tree subgraph and then execute the LineToCompleteBinaryTree subroutine where the edges of the ring subgraph define the line graph for the subroutine and the orientation of the ring defines the parent/children of each node. After the subroutine is finished, we have a spanning complete binary tree subgraph and therefore the final merging is complete. After that, committee leader $v$ also sends the size, the diameter, and $UID_v$ to all followers $x$ in $C(u)$. The messaging part of this algorithm requires $\log d'+\log maxNeighborDiameter$ and the LineToCompleteBinaryTree requires $\log|V(C(u))|)$. 

\smallskip

\noindent\textbf{Trees.} Since $d>O(\log n)$ we cannot use the previous method to form the CompleteBinaryTree. Note here that while $d>O(\log n)$, the distance $b$ of each node from at least $1$ ex-committee leader is $b<O(\log n)$ which we will take advantage of. Since the ex-committee leaders knew that they were on a directed tree subgraph in the graph of committees, once each of them finished with its own merging, it will send an activation message that will be propagated to all nodes. Once a node receives the activation message, it starts executing the asynchronous variation of the LineToCompleteBinaryTree. Since we know that every node has $b<O(\log n)$, every node will start executing the the asynchronous variation of the LineToCompleteBinaryTree after at most $O(\log n)$ rounds. Basically, every node has a different waking-up round between $0$ and $\log n$ The asynchronous LineToCompleteBinaryTree requires $\log n$ rounds after the final node in the line awakens. Therefore this step takes $O(\log n)$ rounds.

\smallskip

\noindent\textbf{Asynchronous LineToCompleteBinaryTree}

\smallskip

Our goal here is to make the nodes simulate the protocol of the synchronous LineToCompleteBinaryTree. The difficulty arises from the fact that nodes wake up at different points. Consider a spanning line of size $n$ with the root being the ``right" endpoint of the line. Let us call the nodes $u_0,u_1,u_2,\ldots,u_{n-1}$  for $j=1,2,\ldots,n-1$ starting from the ``left" endpoint of the line. In order to make sure that asynchronous version works correctly we have to make sure of the following things: (i) Each node $u_j$ never has a degree of more than $4$. This guarantees that the degree of each node stays the same as the synchronous version. (ii) Each node $u_j$ has to activate exactly the following edges  $(u_j,u_{j+2^1}),(u_j,u_{j+2^2}),\ldots,(u_j,u_{j+2^i})$ in this exact order, for $((i=1,2,\ldots,n)\lor (j+2^i<n-1))$ and has to deactivate exactly the following edges $(u_j,u_{j+2^0}),(u_j,u_{j+2^1}),\ldots,(u_j,u_{j+2^{i-1}})$ in this exact order for $(i=1,2,\ldots,n)\lor j+2^{i-1}<n-1$. This guarantees that both versions have the exact same edge activations and deactivations.

The asynchronous LineToCompleteBinaryTree works as follows:

Each node $u$ keeps a counter called $EA=0$ and $EDA=0$. Counter $EA$ tracks how many edges node $u$ has activated and counter $DEA$ tracks how many edges node $u$ has deactivated. Consider node $u$ where its parent is called $v$, its grandparent is called $w$, its child is called $x$ and the root of the tree is called $r$.
In each round $2\cdot i+1$, each node $u$ activates an edge $w$ if $u,v,w$ are awake AND $EA_u=EA_v$ AND $DEA_u=EA_u$. In each round $2\cdot i$, each node deactivates an edge with its parent if $u,v,x$ are awake AND $EA_x=DEA_u+1$ AND $EA_u=DEA_u+1$. In each round $i$, $r$ sends a termination if it has two children. If in some arbitrary round $j$ node $u$ receives a termination message, and in the beginning of round $2\cdot i+1>j$ its grandparent $w$ has $2$ children AND $EA_u=EA_w-1$ then node $u$ enters the termination state where in each round, it sends a termination message to its children. Note here that the two children of $r$ have no grandparent, and they enter the termination state once their parent (which is $r$) has two children.

\begin{lemma}
	Consider two initial line graphs $D(i)=(V,E(i))$ and $D'(i)=(V',E'(i))$ with size $n$, where $V=V'$ and $E(i)=E'(i)$. Executing the asynchronous LineToCompleteBinaryTree algorithm on graph $D(i)=(V,E(i))$ yields the final graph $D(f)=(V,E(f))$ and executing the synchronous LineToCompleteBinaryTree algorithm on graph $D'(i)=(V',E'(i))$ yields the final graph $D'(\log n)=(V',E'(\log n))$, where $E(f)=E'(\log n)$. 
\end{lemma}

\begin{proof}
	
	Consider any arbitrary node $u_j$ on the line $u_0,u_1,\ldots,u_{n-1}$. The condition ($DEA_{u_j}=EA_{u_j}$) for activating an edge imposed by the algorithm does not allow node $u_j$ to have more than $2$ active edges with nodes $u_k$, where $k>j$. The condition $EA_{u_j}=EA_{u_i}$ for activating an edge imposed by the algorithm does not allow node $u_j$ to have more than $2$ active edges with nodes $u_l$, where $l<j$. Therefore node $u_j$ can never have more than $4$ active edges. 
	
	Now let us prove that node $u_j$ always activates and deactivates the edges based on the previous analysis.
	
	Consider $p$ the round in which node $u_j$ has woken up. For every round $g<p$, node $u_{j+1}$ will not deactivate edge $(u_{j+1},u_{j+2})$ by definition of the algorithm, and if node $u_{j-1}$ is not awake node $u_j$ will not deactivate $(u_{j-1},u_j)$.  Once nodes $u_j,u_{j+1},u_{j+2}$, node $u_j$ will activate edge $(u_j,u_{j+2})$ and after this once node $u_{j-1}$ activates $(u_{j-1},u_{j+1})$ node $u_j$ will deactivate $(u_{j-1},u_j)$.
	
	Consider round $2\cdot b>p$. Lets us assume that in rounds $p,p+1,\ldots,b-1$, the correct edges were being activated. This means that node $u_j$ now has an active edge with its current parent $u_{j+2^{b-1}}$. Note here that node $u_j$ will not activate an edge with its grandparent since $EA_{u_j}>EA_{u_{j+2^{b-1}}}$. Once node $u_{j+2^{b-1}}$ activates  $EA_{u_j}$ in total, it will have $u_{j+2^{b}}$ as a parent. Once this happens, node $u_j$ activates an edge with $u_{j+2^{b}}$. 		
	Consider round $2\cdot b+1>p$. Now let us assume that in rounds $p,p+1,\ldots,b-1$, the correct edges were being deactivated. This means that node $u_j$ now has an active edge with its child $u_{j-2^{b-1}}$. Node $u_j$ will not deactivate the edge with $u_{j+2^{b-1}}$ since $EA_{u_{j-2^{b-1}}}<DEA_{u_j}+1$. Once node $u_{j-2^{b-1}}$ activates an edge with $u_{j+2^{b-1}}$, $EA_{u_{j-2^{b-1}}}=DEA_{u_j}+1$ and node $u_j$ will deactivate an edge with with $u_{j+2^{b-1}}$.
	
	Finally we need to prove, that node $u_j$ will terminate and will not keep activating and deactivating edges indefinitely. Note here that node $u_{n-1}$ will definitely terminate at some point since node $u_{n-3}$ will activate an edge with $u_{n-1}$ once it awakens. Each time that a node $u$ terminates, its sends a message to its children to terminate as well. Therefore, we can see that all nodes will eventually terminate. Also note that no node will receive a termination message until its parent terminates as well. The termination state condition can then be translated to: ``if my parent (and my grandparent) has entered termination mode, and my grandparent has $2$ children, i will also enter the termination mode". Note that his can only happen when node $u$ has reached its final position in the CompleteBinaryTree. Node $u$ can never activate an edge after it reaches its final position since its grandparent $w$ will always have two children.
\end{proof}

\begin{corollary}
	Consider an initial line graph $D(i)=(V,E(i))$ with $n$ nodes. We execute the asynchronous LineToCompleteBinaryTree algorithm on the graph. Assume that the final node on the graph wakes up in round $k$. The LineToCompleteBinaryTree algorithm requires $O(\log(n)+k)$ rounds to terminate.
\end{corollary}

\begin{proof}
	
	Note that after all nodes have awoken, the algorithm requires $O(\log n)$ rounds to terminate. Consider node $u_j$ that was one of the last nodes to wake up. In each odd round after waking up, node $u_j$ will activate an edge and in each even round after waking up, it will deactivate an edge. This happens because every neighbor $u_k$ of $u_j$ is awake and it has $EA_{u_k}\geq EA_{u_j}$ and $EDA_{u_j}\geq EDA_{u_i}$.
\end{proof}

\section{Omitted Details from Section \ref{SEC:ALGORITHM-3} - Trading the Degree for Time}

Since the high level strategy of our algorithm is exactly the same as the previous ones apart from minor differences, we are just going to list how we handle the new differences instead of the whole algorithm. First consider the selection graph. 

\begin{enumerate}

\item Each committee $u$ that has more than 1 child prepares to deactivate all of its edges with its children in the selection graph apart from the one with the greatest UID among its children. Note now that if the deactivations do happen, in the selection graph we will have subgraphs that are either single committees, or pairs of committees or lines of committees. Thus, we have managed to get rid of the directed trees. The problem now is that we have a lot of single committees. 

\item Each committee $u$ performs a matching between its children that are leaves. If they are an odd number of children, then $u$ matches itself with one if its children as well so that everyone is matched. Thus, we have also managed to minimize the number of single committees.

\item Each committee $u$ deactivates the edges listed on step 1. 

\end{enumerate}

Currently in the selection graph, we have pairs of committees and lines of committees. Each pair merges in a single committee in $1$ phase and each committee on a line execute a variant of the LineToCompleteBinaryTree which requires $O(\log k)$ rounds where $k$ is the diameter of the line. We have now finished with the high level strategy of the algorithm. 

We are now going to provide the low level details of our algorithm which are the ones that allow us to minimize the running time of our algorithm. We are going to omit the communication description since it is very similar to the GraphToWreath algorithm. 

\textbf{Leader Merging Mode}

Consider two committees which are $C(p)=(K,L(p),u)$ and $C'(p)=(K',L'(p),v)$ and committee $u$ decides to merge with committee $v$ through nodes $x\in K$ and $y\in K'$. They both agree on the merging (identical to how the GraphToWreath committees communicate), first they merge their ring components (identical to how the Wreath committees merge their ring component) and then node $x$ activates edges in $C(p)$ until it activates an edge with leader $u$ and node $y$ activates edges in $C'(p)$ until it activates an edge with leader $y$. After that, $u$ activates an edge with $y$ and the merging is complete. We now have the new committee  $C''(p)=(K'',L''(p),v)$. Let us assume that prior to the merging the shortest path between leader $u$ and every other node in committees $C(p)$ is $O(\frac{\log n}{\log\log n})$. Let us assume the same for committee $v$. Then after the merging, the same is true for leader $v$ and every other node in committee $C''(p)$.

\textbf{Asynchronous LineToCompletePolylogarithmicTree}

Every leader starts executing a variant of the asynchronous LineToCompleteBinaryTree. The difference of this variant algorithm called LineToCompletePolylogarithmicTree, is that the criteria for entering the termination stage is that your grandparent has $\log n$ children instead of $2$. Every node knows the size $n$ of the network, and therefore knows the $\log n$ upper bound needed for termination. Note here that once all of the leaders finish the asynchronous LineToCompletePolylogarithmicTree the shortest path between any $2$ leaders is $O(\frac{\log n}{\log\log n})$. Let us assume that prior to the merging the shortest path between leader $u$ and every other node in committees $C(p)$ is $O(\frac{\log n}{\log\log n})$. Let us assume the same for all committees on the line. Then after the merging, the same is true for leader $v$ and every other node in the final committee.

\textbf{Matching}

We still have to define how the matching is done between the children of each committee $u$ . The difficulty here is that after we do the matching, the children have to become neighbors through $u$ and we have to handle this properly so that we can guarantee that the degree of each node stays polylogarithmic. In order to achieve this, each leader $u$ keeps a virtual addressing of each node in its committee on the leaves of its committee. Since $n/2$ nodes of the committee are leaf nodes, we can achieve this by addressing two nodes to each leaf node. Consider multiple neighboring committees that committee $u$ has to match. Committee $u$ computes the minimum distance between all pair of committees and then matches the pair the minimum distance. Committee $u$ keeps recomputing and matching until all pairs of nodes are matched. 
Note here that after each merging that committee $u$ does, it has to virtually readdress all the nodes in its committee based on the new polylogarithmic tree.

\section{Omitted Details from Section \ref{SEC:LOWER-BOUNDS} - Lower Bounds for the Depth-$\log n$ Tree Problem}
\label{appSEC:LOWER-BOUNDS}

We will now shift our focus into proving lower bounds for our model. We are going to provide lower bounds for both a centralized model and a distributed one because we want to show that there is an important difference between the two of them. 

\subsection{Centralized Setting}

In the centralized setting, everything we have previously defined in the model subsection stays the same but now every node also has complete knowledge of the graph and a centralized controller can decide what each node will do in each round.

\begin{definition}
	We define the potential of a $UID_{u}$ to $v$ as its minimum "distance" from $v$. The distance is defined as follows: Consider all nodes $w$ in the network that know $UID_u$. Compute the length of the shortest path between each node $w$ and node $v$. The minimum length among all shortest paths is the distance between $UID_u$ and node $v$. We denote the potential of $UID_{u}$ to $v$ by $PO_{u,v}$. 
\end{definition}

Note that in any initial graph $D=(V,E)$, $\forall u,v \in V, PO_{u,v}= \max\limits_{u} PO_{u,v}=n-1$. Consider any pair of nodes $u,v$, where $PO_{u,v}=k$. There are two ways to reduce $PO_{u,v}$ in each round $i$:

$\bullet$ \textbf{Information Propagation}. Consider all nodes $w$ that currently know $UID_u$. Compute the shortest path between all pairs of $w$ and $v$ and pick node $w$ that yields the smallest shortest path. Node $w$ can send to $UID_u$ one of its neighbors $y$ that belong on the shortest path between $w$ and $v$ to reduce $PO_{u,v}$ by $1$.  

$\bullet$ \textbf{Reduce Shortest Paths}. Consider all nodes $w$ that currently know $UID_u$. Compute the shortest path between all pairs of $w$ and $v$ and pick node $w$ that yields the smallest shortest path with $size=k$. Now consider all pairs of nodes $x,y$ that are potential neighbors and also belong on the shortest path between $w$ and $v$. Activating $xy$ between one pair of $x,y$ reduces $PO_{u,v}$ by $1$. Activating multiple $xy$ between different pairs in one round can reduce $PO_{u,v}$ even more but at most by $k/2$.

\begin{observation} \label{appPotential}
	In order for an algorithm to solve the Depth-$\log n$ Tree Problem, $\forall u,v\in V, PO_{u,v}\leq \log n$.
\end{observation}

\begin{lemma} \label{appLowerTimeBound}
	Any transformation strategy based on this model requires $\Omega(\log n)$ time to solve the Depth-$\log n$ tree problem if the initial graph $G_s$ is a spanning line.
\end{lemma}

\begin{proof}
	
	Consider a spanning line where, for simplicity, we call the node that resides at the ``left" endpoint of the line $u$ and the node that resides at the ``right" endpoint of the line $v$. According to observation \ref{appPotential}, in order for an algorithm to solve the Depth-$\log n$ tree problem, $PO_{u,v}\leq \log n$. In the initial graph, $PO_{u,v}=n-1$. 
	We know that by using Edge Activations, we can reduce $PO_{u,v}$ by half in each round, and by using Information Propagation we can reduce $PO_{u,v}$ by 1 in each round. Therefore in order for $PO_{u,v}=\log n$, any algorithm would required at least $\Omega(\log n)$ rounds.
\end{proof}

\begin{lemma}\label{appEdgeActivations}
	Any transformation strategy based on this model that solves the Depth-$\log n$ Tree problem in $O(\log n)$ time, requires $\Omega(n)$ edge activations.
\end{lemma}

\begin{proof}
	Let us again consider a spanning line as the initial graph. W.l.o.g. let us assume that the size of the network is odd. Let us call $u$ the node that is the ``left" end point of the line and $v$ the ``right'' endpoint of the line.
	
	Let us assume that in some round $i$, where $i\leq \log n$, that $PO_{u,v}\leq \log n$. We can produce the following equation based on the two rules that allow us to reduce the potential: $InitialPotential-\#EdgeActivations-\#MessagesSent\leq \log n$. The maximum value of $MessagesSent$ is $\log n$ and $InitialPotential=n-1$ and if we add those in the previous equation we get  $\#EdgeActivations\geq n-1-2\log n$ and therefore, in order for $PO_{u,v}\leq \log n$ at least $n-1-2\log n$ edges have to have been activated.
\end{proof}

\begin{lemma}
	Any transformation strategy based on this model that solves the Depth-$\log n$ Tree problem in $O(\log n)$ time, requires $\Omega (n/\log n)$ edge activations per round.
\end{lemma}

\begin{proof}
	
	From Lemma \ref{appEdgeActivations} we know that $PO_{u,v}=0$ to be possible in $\log n$ time, the following equation has to be true $EdgeActivations=\Omega(n)$. Now, since we are trying to find the minimum number of edge activations per round possible, we can easily do this by dividing the total number of edge activations with the number of rounds. Therefore $EdgeActivationsPerRound\geq \frac{EdgeActivations}{Rounds} \geq \frac{\Omega(n)}{\log n}$. 
\end{proof}

Since we have have just proven that $\Omega(n)$ edge activations are required in order to solve the Depth-$\log n$ problem given any initial graph, we are now going to prove that $\Theta(n)$ edges are sufficient in order to solve it. First, we are going to informally prove it for the special case of the spanning line graph and afterwards we are going to prove it for general graphs.

Consider a spanning line with nodes $u_1,u_2,\ldots,u_j$ for $j=1,2,\ldots,n$. For simplicity, assume that $u_1$ is the ``left" endpoint of the line, $u_2$ is the neighbor of $u_1$ etc, $u_3$ is a neighbor of $u_2$ etc. In each round $i$, we activate edge $u_j,u_{j+2^{i}}$ $\forall$ $\{u_j|(j\mod(2^i)=1)\land (j+2^{i}\leq n)\}$. After $\log n$ rounds, the diameter of the shape is equal to $\log n$. Let us now proceed to analyzing the total edge activations. By definition of the algorithm, in each round $i$, $\frac{n}{2^i}$ edges are activated. Since the algorithm runs for $\log n rounds$, we have $\sum_{i=1}^{\log n} \frac{n}{2^i}=n$ total edge activations. We call this algorithm CutInHalf.

\begin{theorem}
	Given any initial graph $D=(V,E)$, the Depth-$\log n$ problem can be solved in $O(\log n)$ time, with $\Theta(n)$ total edge activations.
\end{theorem}

\begin{proof}
	Since we are in a centralized setting, we are first going to perform some global computations that are going to output the specific edges that have to be activated in order for the diameter of the shape to drop to $\log n$.
	We consider any initial graph $D=(V,E)$ and we pick an arbitrary node called $u$. First, we compute a spanning tree that starts from node $u$. Afterwards we compute an eulerian tour starting from $u$. This way we can create a virtual ring $D'=(V',E')$ that has $|V'|\leq|2\cdot V|$ and $|E'|\leq 2|E|$. Now in this ring, node $u$ deactivates one of its incident edges and the graph is now a line. We can now execute the CutInHalf algorithm to solve the Depth-$\log n$ Tree problem in $O(\log n)$ time, with $\Theta(n)$ total edge activations .
\end{proof}

\subsection{Distributed Setting}

In this part, we are going to show that there is a difference in the minimum total edge activations required for solving the Depth-$\log n$ problem between the centralized and the distributed model.

\begin{definition}
	Let $U={u_1,u_2,\ldots,u_k}$ be a sequence of UIDs of length $k$. We say that $U$ is an increasing order sequence if, for all $i,j,1\leq i,j\leq k$, we have $i\leq j$ iff $u_i\leq u_j$. 
\end{definition}

\begin{definition}
	Let $A$ be a comparison-based algorithm executing on an increasing order ring graph. Let $i$ and $j$ be two nodes in the ring graph. We say that $i$ and $j$ are in corresponding states if the UIDs that they both have received from counterclockwise neighbors are a decreasing order sequence and the UIDs they have received are an increasing order sequence and vice versa.
\end{definition}

\begin{definition}\label{appIncreasing Order Ring}
	We define the increasing order ring $R$ as follows. Suppose we have an increasing order sequence $U$ of UIDs to be assigned on a ring with $n$ nodes. We assign the smallest UID from $U={u_1,u_2,\ldots,u_k}$ to an arbitrary node and we continue assigning increasing UIDs clockwise (or counterclockwise). We call this an increasing order ring.
\end{definition}

\begin{definition}
	We define a round of an execution/algorithm to be active if at least one message is sent in it or an edge is activated in it.	
\end{definition}

\begin{definition}
	We define the \emph{k-expo-neighborhood} of node $i$ in ring $R$ of size $n$, where $0\leq k \leq \frac{n}{2}$, to consist of the $2\cdot 2^k+1$ nodes $i-2^k,\ldots,i+2^k$, that is, those that are within distance at most $2^k$ from node $i$ (including $i$ itself).
\end{definition}	

\begin{lemma} \label{appCorresponding States}
	Let $A$ be a comparison-based algorithm executing in an increasing order ring of size $n$ and let $d_{min}$ be the initial distance from between node $d$ and the node with the minimum UID called $d_0$ and $d_{max}$ be the initial distance from between node $d$ and the node with me maximum UID called $d_{n-1}$. Let $i$ and $j$ be two nodes in $A$ where $2^k=\min(\max(i_{min},i_{max}),\max(j_{min},j_{max}))$. Then, at any point after at most $k$ active rounds, nodes $i$ and $j$ are in corresponding states, with respect to the UID sequences.
\end{lemma}

\begin{proof}
	Note here that nodes $i$ and $j$ are in corresponding states as long as $(((PO_{d_0,i}>0) \lor (PO_{d_{n-1},i}>0)) \land  ((PO_{d_0,j}>0) \lor (PO_{d_{n-1},j}>0))$. In simple terms, $i$ and $j$ are in corresponding states as long as both of them do not know both $UID_{d_0}$ and $UID_{d_{n-1}}$.
	Consider $p=\max(PO_{d_0,i},PO_{d_{n-1},i})$ prior to the execution of the algorithm. We know that the initial distance between $d_0,i$ and $d_{n-1},i$ is at least $2^k$ and therefore $p>2^k$. We already know from a previous proof that at least $k-\log k$ are needed before $((PO_{d_0,i}>0) \lor (PO_{d_{n-1},i}>0))$. A similar argument is used for $j$.
\end{proof}

\begin{observation} \label{appActive rounds}
	Any transformation strategy based on this model that solves the Depth-$\log n$ Tree problem in $O(\log n)$ time in an increasing order ring, requires at least $\log n$ active rounds.
\end{observation}

\begin{theorem}
	Any distributed algorithm that solves the Depth-$\log n$ Tree problem in $O(\log n)$ time, requires $\Omega (n\log n)$ total edge activations.
\end{theorem}

\begin{proof}
	Consider an increasing order ring $R$ with $n$ nodes and algorithm $A$ that solves the Depth-$\log n$ problem. Consider the node with the greatest UID in the network, called $u_{max}$, the node with the smallest UID in the network, called $u_1$, and the antipodal node of $u_{max}$ called $u_c$.
	
	First of all, note that in the first round, all nodes except from $u_1$ and $u_{max}$ are in corresponding states. We can generalize this statement by using Lemma \ref{appCorresponding States} to state that in round $i$, each node whose i-expo-neighborhood does not include both $u_1$,$u_{max}$ is in a corresponding state with each such node. Therefore those nodes behave the same way e.g. if in round $i$, one of those $c$ nodes activates an edge, then all $c$ nodes activate an edge. For this proof, we define a round of algorithm $A$ to be live if the $c$ nodes activates at least one edge in it, we also define a round of algorithm $A$ to be asleep if none of the $c$ nodes activate an edge in it.
	
	We already know that we need at least $\log n$ active rounds to connect $u_{max}$ with $u_c$ from Lemma \ref{appActive rounds}. Our goal here is to prove that $\log n$ of those active rounds also have to be live rounds. 
	
	For simplicity, we define the set $C$ where node $u\in C$ if $u$ is in the same corresponding state as $u_a$ (including $u_a$), the set $A$ where node $u\in A$ if $u$ is not in the same corresponding state as $u_c$. 
	
	Consider an arbitrary round $i$, where the shortest path between $u_{max}$ and $u_c$ is $|P|=k$. This shortest path can be split into two different paths. The one called $P_A$ that includes nodes $u\in A$ and the one called $P_C$ that includes nodes $v\in C$. Essentially, the potential $PO_{u_{max},a}=|P_C|$. Let us divide our analysis between asleep and live rounds and study how much the potential can be reduced in each round.
	
	$\bullet$ \textbf{Asleep rounds.} In each asleep round $a$, only nodes $u\in A$ can activate edges and $|P_C|$ can only be reduced by at most $l+1$ where $l$ is the total number of live rounds before round $a$. We can reduce it $l$ by having $u\in A$ activating an edge with each potential neighbor $v\in C$, and reduce it by $1$ by having $u$ send $UID_{u_{max}}$ to all $v\in C$.
	
	$\bullet$ \textbf{Live rounds.} In each live round $l$, all nodes can activate an edge so we can reduce $|P_C|$ by $l+1$ by following the above strategy and additionally, use edge activations between nodes $v\in C$ so that $|P_C|$ is reduced by at most half.
	
	Note here, that Asleep rounds are not enough to reduce the potential to $0$ in order to solve the Depth-$\log n$ problem. After $O \log (n)$  asleep rounds, $PO_{u_{max},a}\geq InitialPotential-(\log n)(l+1)= \frac{n}{2}-(\log n)(l+1)$. Therefore we need at least $\log n$ live rounds to solve the Depth-$\log n$ problem.
	
	We are now gonna examine how many edges are activated in each live round. Recall that in each live round $l$ , at least $1$ node $v \in C$ activates an edge and by Lemma \ref{appCorresponding States}, all nodes $v\in C$ activate an edge. The number of nodes $v\in C$ in round $i$ are  $|u|\geq\#CnodesInInitialGraph-NodesRemovedInPreviousLiveRounds-$ $NodesRemovedInPreviousAsleepRounds$ $=(n-2)-(\sum_{i=1}^{l-1} 2^i) (\sum_{i=1}^{a} -i(l-1))-a(l-1)$.
	The number of edges activated in each round $l$ are at least $|C|\geq |u|.$
	Therefore the total number of edge activations in live rounds after $\log n$ rounds is at least $(n-2)-(\sum_{i=1}^{\log n} 2^i) (\sum_{i=1}^{\log n} -i(l-1))-a(l-1)=\Theta (\log n)$
\end{proof}

\end{document}